\documentclass[11pt,a4paper]{amsart}
\usepackage[left=3.5cm, top=3cm, bottom=3cm, right=3.5cm]{geometry}
\usepackage{amsmath,amsfonts,amssymb,amsthm,amscd,graphicx}
\pdfoutput=1

\makeatletter
\newcommand\xleftrightarrow[2][]{%
  \ext@arrow 9999{\longleftrightarrowfill@}{#1}{#2}}
\newcommand\longleftrightarrowfill@{%
  \arrowfill@\leftarrow\relbar\rightarrow}
\makeatother

\newtheorem{thm}{Theorem}
\newtheorem{lem}[thm]{Lemma}

\newtheorem{defn}[thm]{Definition}

\newtheorem{example}[thm]{Example}

\def\res{\mathop{\rm res}\limits}

\begin{document}

\title[Extended $\bigvee$-systems and extended affine Weyl orbit spaces]{Extended $\bigvee$-systems and almost-duality for extended affine Weyl orbit spaces}
\date{\today}

\author{Richard Stedman and Ian A.B. Strachan}

\address{School of Mathematics and Statistics\\ University of Glasgow\\Glasgow G12 8QQ\\ U.K.}
\email{richard.stedman@gmail.com}
\email{ian.strachan@glasgow.ac.uk}

 \subjclass{}

\maketitle

\noindent {\bf Abstract:} Rational solutions of the Witten-Dijkgraaf-Verlinde-Verlinde (or WDVV) equations of associativity are given in terms a configurations of vectors which satisfy certain algebraic conditions known as $\bigvee$-conditions \cite{V}. The simplest examples of such configuration are the root systems of finite Coxeter groups. In this paper conditions are derived which ensure that an extended configuration - a configuration in a space one-dimension higher - satisfy
these $\bigvee$-conditions. Such a construction utilizes the notion of a small-orbit, as defined in \cite{S}. Symmetries of such resulting solutions to the WDVV-equations are studied; in particular, Legendre transformations. It is shown that these Legendre transformations map extended-rational solutions to trigonometric solutions and, for certain values of the free data, one obtains a transformation from extended $\bigvee$-systems to the trigonometric almost dual solutions corresponding to the classical extended affine Weyl groups.

\tableofcontents

\newpage

\section{Introduction}

\subsection{Background}

Solutions of the form
\begin{equation}
F=\frac{1}{4} \sum_{\alpha\in\mathcal{U}} h_\alpha \alpha(z)^2 \log \alpha(z)
\label{rational}
\end{equation}
of the Witten-Dijkgraaf-Verlinde-Verlinde (or WDVV) equations have appeared in many different settings, for example:
\begin{itemize}
\item[-] in Seiberg-Witten theory \cite{MMM};
\item[-] in the theory of almost-dual Frobenius manifolds \cite{D2};
\item[-] on the space of stability conditions \cite{B}.
\end{itemize}
Here $\mathcal{U}$ is a finite collection of covectors and $h_\alpha$ are numbers - known as multiplicities - associated to each covector. In the first examples of this type, the sets $\mathcal{U}$ were the collections $\mathcal{R}_W$ of root vectors of a finite Coxeter groups $W$ and in order to understand and generalize the construction of these solutions Veselov introduced the notion of a $\bigvee$-system \cite{V}. Roughly (and a precise definition will follow), in order for (\ref{rational}) to satisfy the WDVV equations (which are a set of over-determined partial differential equations) the covectors in
$\mathcal{U}$ have to satisfy certain {\sl algebraic} conditions, and if these are satisfied, the configuration is said to be a $\bigvee$-system.
If the multiplicities are positive they may be absorbed into the covectors, but if some are negative such a procedure breaks the geometry - or rather -  moves it into the complex, so in this paper we keep the multiplicities separate. Thus for us a $\bigvee$-system consists of {\sl pairs} $(\alpha\,,h_\alpha)$, consisting of a covector and its associated multiplicity, which satisfy the $\bigvee$-conditions (which are given in Definition \ref{Vdef}).

The origins of this paper was to understand the geometry of the configurations which define certain solution of the WDVV equation of almost-dual type. Starting with a superpotential (see Section \ref{hurwitz})
\begin{equation}
\lambda(z) = \left.\prod_{i=0}^N \left(z-z^i \right)^{k_i} \right|_{\sum_{i=0}^N k_i z^i=0}
\label{genA}
\end{equation}
one obtains the solution
\[
F=\frac{1}{4}\left. \sum_{i\neq j} k_i k_j (z^i-z^j)^2 \log(z^i-z^j) \right|_{\sum_{i=0}^N k_i z^i=0}
\]
with metric $g=\left.\sum_{i=0}^N k_i d(z^i)^2\right|_{\sum k_i z^i=0}$
and from the superpotential
\begin{equation}
\lambda(z) = z^s \prod_{i=1}^N \left(z^2-(z^i)^2 \right)^{k_i}
\label{genB}
\end{equation}
one obtains the solution
\[
F= \sum_{i=1}^N k_i (s+ 2 k_i) (z^i)^2 \log z^i + \frac{1}{2} \sum_{i\neq j} k_i k_j \left( z^i \pm z^j\right)^2 \log \left( z^i \pm z^j\right)\,.
\]
with metric $g=\sum_{i=1}^N k_i d(z^i)^2$ (note, this metric is independent of the value of $s$).
In the cases when $k_i=1$ for all $i$ these are the well-known $A_N$ and (for $s=0$) $B_N$ solutions, with a configuration being the root vectors of these Coxeter groups. Introducing multiplicities, either in the zeros or poles of the superpotential - depending on the sign of the $k_i$ - destroys this interpretation, and also introduces a split signature metric. From the analysis of the case when
\[
{\bf k}^{ext} = \{k\,, \underbrace{1\,,\ldots\,,1}_m\,,\underbrace{-1\,,\ldots\,,-1}_n\}
\]
it was found that the configuration could be interpreted as an extension into a perpendicular direction of the lower dimensional configuration
defined by
\[
{\bf k}= \{ \underbrace{1\,,\ldots\,,1}_m\,,\underbrace{-1\,,\ldots\,,-1}_n\}\,.
\]
The origin of this configuration - which defines a generalized root system - from a {\sl rational} superpotential aids in the interpretation of its symmetries: the superpotential is invariant under interchange of zeros and of poles. Isotropic roots can be interpreted as an interchange of zeros and poles

\subsection{Outline}

The purpose of this paper is three-fold. Firstly we construct extended $\bigvee$-systems. Starting with a $\bigvee$-system we extend the configuration into a one-dimension higher space by adding a one-dimensional orthogonal direction and adding certain special covectors to the original configuration. We then derive the (extra)-conditions required for this extended configuration to be a $\bigvee$-system. This construction utilizes the idea of a small-orbit, as introduced by Serganova \cite{S}. Thus extended $\bigvee$-systems are $\bigvee$-systems, but in one dimension higher than the original system.

Secondly, we perform a Legendre transformation on an extended $\bigvee$-system. Such Legendre transformations are symmetries of the WDVV-equations, and hence map solutions to solutions \cite{D1}. To perform such a transformation requires the choice of a direction, and for extended $\bigvee$ there is a natural choice of direction, namely the newly introduced orthogonal direction perpendicular to the original space. Using this direction for the Legendre transformation results in a transformation that maps rational solutions, i.e. those of the form (\ref{rational}), to trigonometric systems, that is, to solution of the form
\begin{equation}
F={\rm cubic~} + \sum_{\alpha \in \mathcal{U}} h_\alpha Li_3 (e^{\alpha(z)})\,
\label{trigonometric}
\end{equation}
where $Li_3(z)$ is the tri-logarithm function. A separate theory of trigonometric $\bigvee$-systems has been developed by M.Feigin \cite{F}.

Finally, we make the connection between extended $\bigvee$-systems and the almost-dual Frobenius manifolds for the extended affine Weyl group orbit spaces as constructed and studied in \cite{DZ} and \cite{DSZZ}. In particular the following is proved:

\begin{thm}
Let $W$ be a finite irreducible classical Coxeter group of rank $l$ and let $\widetilde{W}$ be the extended affine Weyl group of $W$ with arbitrary marked node. Then up to a Legendre transformation, the almost dual prepotentials of the classical extended affine Weyl group orbit spaces
$\mathbb{C}^{l+1}/\widetilde{W}$ are, for specific values of the free data, the extended $\bigvee$-systems of the $\bigvee$-system $R_W\,.$
\end{thm}

\noindent A more precise version of the Theorem - with the full data specified - will be given in Section 5.

\subsection{$\bigvee$-systems and small orbits}

We start by defining a $\bigvee$-system and the $\bigvee$-conditions.
Let $V$ be a real vector space and $\mathcal{U} \subset V^*$ a finite set of covectors which span the dual space $V^*\,.$
With this one defines a metric by the formula
\begin{equation}
\sum_{\alpha \in \mathcal{U}} h_\alpha \alpha(x) \alpha(y) = h_\mathcal{U} \, (x,y)\,.
\label{canonicalmetric}
\end{equation}
Here $h_\alpha$ is the multiplicity of the covector $\alpha$, which may be negative or non-integer. It is therefore an additional assumption that the metric
$(x,y)$ defined in this way is non-degenerate: it certainly can be non-positive definite (see example \ref{basicexample}).
The constant $h_\mathcal{U}$ is technically redundant - it could be absorbed into the multiplicities; however, it is useful to keep this freedom.

\begin{example}
With the standard representation of the roots and metric for the Coxeter group $A_n$,
\begin{eqnarray*}
\mathcal{U} & = & \{ e_i - e_j \,, 0 \leq i\,,j \leq n\,, i \neq j \}\,,\\
g & = & \left.\sum_{i=0}^n \left(dz^i\right)^2\right|_{\sum z^i=0}\,,
\end{eqnarray*}
one has
\[
\sum_{\alpha\in\mathcal{U}} \alpha(x) \alpha(y) = 2 (n+1) (x,y)
\]
so one takes the normalization $h_\alpha=1$ (the roots form a single orbit) which forces $h_\mathcal{U} = {\rm (dual)~Coxeter~number}\,.$
\end{example}

Again following \cite{V} this metric defines the isomorphism $\varphi_{\mathcal{U}}:V \rightarrow V^*$ and we denote $\varphi_\mathcal{U}^{-1}(\alpha) = \alpha^\vee\,.$

\begin{defn}\label{Vdef}
The system $\mathcal{U}$ is a $\bigvee$-system if the following relations are satisfied:
\[
\sum_{\beta \in \Pi \cap \mathcal{U}} h_\beta \beta(\alpha^\vee) \beta^\vee = \lambda \alpha^\vee
\]
for each $\alpha \in \mathcal{U}$ and each two-plane $\Pi \subset V^*$ containing $\alpha$. The constant $\lambda$ will in general depend on both $\Pi$ and $\alpha\,.$
\end{defn}

In the following example the multiplicities are both positive and negative, and this drastically effects to geometry - the natural metric is no longer positive definite. The example also shows how one can interpret the $\bigvee$-system: one can think of it as the extension of a one-dimensional $\bigvee$-system associated to the Coxeter group $A_1$ into two dimensions.

\begin{example}\label{basicexample}
 Consider the following solution to the WDVV equations \cite{RS},
 \[
 F^\star = \frac{1}{4} \left( x^2 \log x + y^2 \log y - (x-y)^2 \log(x-y) \right)
 \]
 with $g=2 dx\,dy\,.$ This solution is the almost dual to the Frobenius manifold defined by the prepotential
 \[
 F=\frac{1}{2}t_1^2 t_2 + t_2^2 \log t_2\,.
 \]
 The configuration of vectors
 $\{ \pm (1,0)\,,\pm (0,1)\,, \pm (1,-1)\}$ seems somewhat asymmetric.

\begin{figure}[h!]
  \includegraphics[width=4cm,angle=90]{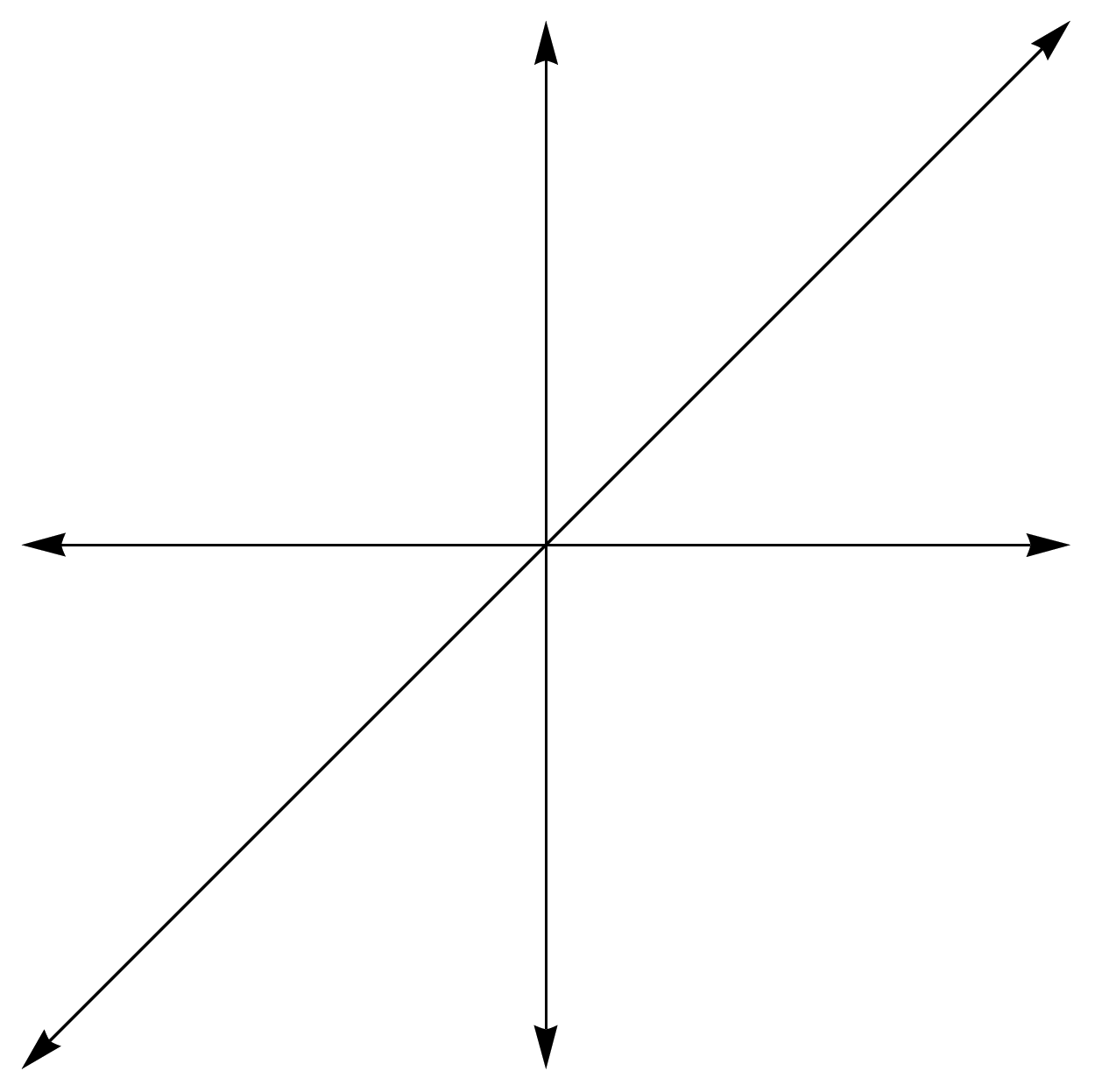}
  \caption{The geometry of the $\bigvee$-configuration}
  \label{fig1}
\end{figure}

\noindent However this does not take into account the split signature of the metric. If one rotates the diagram and superimposes the light-cone, this illuminates the geometry of the configuration:

 \begin{figure}[h!]
  \includegraphics[width=5cm]{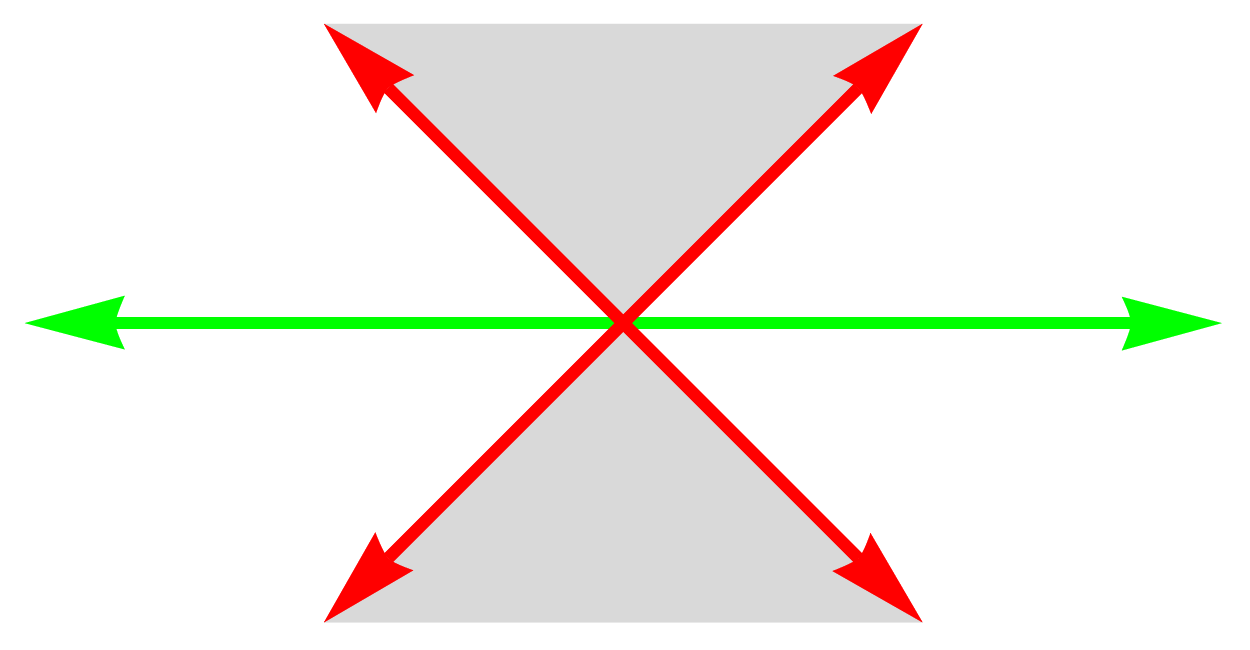}
  \caption{The geometry of the $\bigvee$-configuration, rotated with superimposed light cone}
  \label{fig1}
\end{figure}

 \noindent The vectors $\{ \pm (1,0)\,, \pm (0,1) \}$ are null and the vectors $\pm(1,-1)$ are spacelike.

 \end{example}

This example also motivates the main construction. We start with a configuration $\mathcal{U}=\{\pm(1,-1)\}$ spanning a space $V$ (i.e. the configuration is the root system for the Coxeter group $A_1$) and extend into a perpendicular
direction defined by the normal vector $n=\frac{1}{2} (1,1)\,.$ With this one obtains an extended space $V^{ext}\,,$ and an extended space of configurations $\mathcal{U}^{ext}$
 may be constructed by extending {\sl certain} vectors into the perpendicular direction, so
 \begin{eqnarray*}
 (0,1) & = & \frac{1}{2} (-1,1) + n \,,\\
 (1,0) & = & \frac{1}{2} (1,-1) + n\,.
 \end{eqnarray*}
What makes the vectors $\pm\frac{1}{2} (1,-1)$ special is the following important property: their difference lies in $\mathcal{U}\,.$ This is known as the small orbit property
and its existence will be crucial to what follows.

Thus given a space $V$ and configuration $\mathcal{U}$ we extend into a perpendicular space $V^{ext}$ and form a new configuration $\mathcal{U}^{ext}$ by extending vectors
in the small orbit of $\mathcal{U}$. Imposing the $\bigvee$-conditions on the extended configuration then constrains the various objects, notably the constants $h_a$ associated to each covector.

 \begin{defn} Let $\mathcal{U}$ be a $\bigvee$-system.

 \begin{itemize}

 \item[(a)]
 A {\sl small orbit} $\vartheta_s$ of the $\bigvee$-system is a finite set of covectors such that
 \[
 w_1-w_2 \in \mathcal{U}
 \]
 for all $w_1\,,w_2\in \vartheta_s\,$ with $w_1 \neq w_2\,;$

 \item[(b)] An invariant small orbit consists of pairs $(w,h_w)$, where $w$ is a small orbit covector with associated multiplicity $h_w$, which satisfies the additional conditions:
 \begin{itemize}
 \item[(i)] $\sum_{w \in \vartheta_s} h_w w(z)^2 = h_s (z,z)\,,$
 \item[(ii)] $\sum_{w\in \vartheta_s} h_w w(z) =0$
 \end{itemize}
 for all $z \in V\,.$
 \end{itemize}
 \end{defn}

 \noindent The first part of the definition is just an adaption of the concept of a small orbit for Weyl groups \cite{S}, and the adjective \lq orbit\rq~reflects this origin. In applications this set could be invariant under the action of the Weyl group (and hence a bone-fide orbit), but even if there is no such group we keep this adjective. In the second part of the definition, the adjective \lq invariant\rq~is used since, if $\mathcal{U}$ is a Coxeter configuration (i.e. the root system of a Coxeter group, with multiplicities equal on each orbit), the two conditions in part (b) are, by basic properties of invariant theory, automatically satisfied.

 \medskip

Small orbits for Coxeter configurations were introduced and classified by Serganova \cite{S}:

\begin{thm}\label{Serg} Let $\omega_i$ be the fundamental weight vectors for a finite Coxeter group $W\,.$ The small orbits of rank $\geq 2$ are given by:
\begin{enumerate}
\item $A_n\,: \vartheta_s = W \omega_1 {\rm ~or~} W \omega_n\,;$
\item $B_n\,,BC_n\,,C_n (n\geq 4)\,: \vartheta_s=W \omega_1\,;$
\item $B_2\,: \vartheta_s = W \omega_1\,;$
\item $B_3\,: \vartheta_s = W \omega_1 {\rm ~or~} W \omega_3\,;$
\item $BC_2\,: \vartheta_s = W \omega_1 {\rm ~or~} 2W \omega_2\,;$
\item $C_3\,, BC_3 \,: \vartheta_s = W \omega_1\,;$
\item $D_n (n\geq3\,, n\neq 4)\,: \vartheta_s = W \omega_1\,;$
\item $D_4\,; \vartheta = W \omega_1\,, W\omega_3 {\rm~or~} W\omega_4\,;$
\item $G_2\,: \vartheta_s = W \omega_1\,.$
\end{enumerate}
\end{thm}

\noindent It is interesting to note that the exceptional Coxeter groups do not have any small orbits. Since the constructions in this paper rely on the existence of such small orbits they do not apply to all of the exceptional groups.

The main construction in the first part of this paper may be explained by the following example.

\begin{example}

We begin with the root system (automatically a $\bigvee$-system) for the Coxeter group $A_2\,.$ This is shown in the left-hand diagram in Figure \ref{fig:A2}. The small-orbit is given, by Theorem \ref{Serg}, by the orbit of a weight vector, and this orbit is shown, superimposed on the root system, in the right-hand diagram in Figure \ref{fig:A2}.

\begin{figure}[h!]
  \includegraphics[width=8cm]{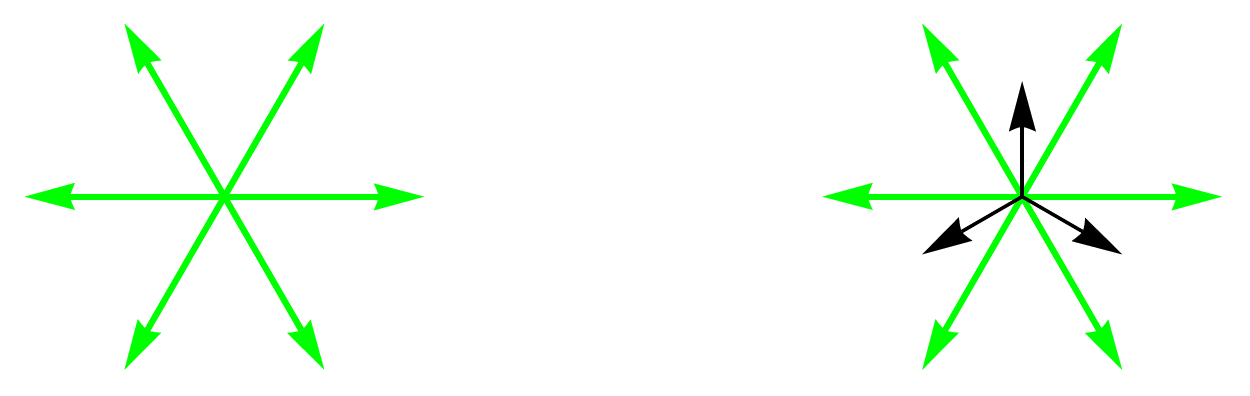}
  \caption{The roots of $A_2$ and the small orbit superimposed on the roots}
  \label{fig:A2}
\end{figure}

\noindent We now extend the configuration into a third dimension by adding a normal vector to the end of each small-orbit vector and adding its negative. For  $A_2$ this is shown in Figure \ref{fig:fig2}.

\begin{figure}[h!]
  \includegraphics[width=8cm]{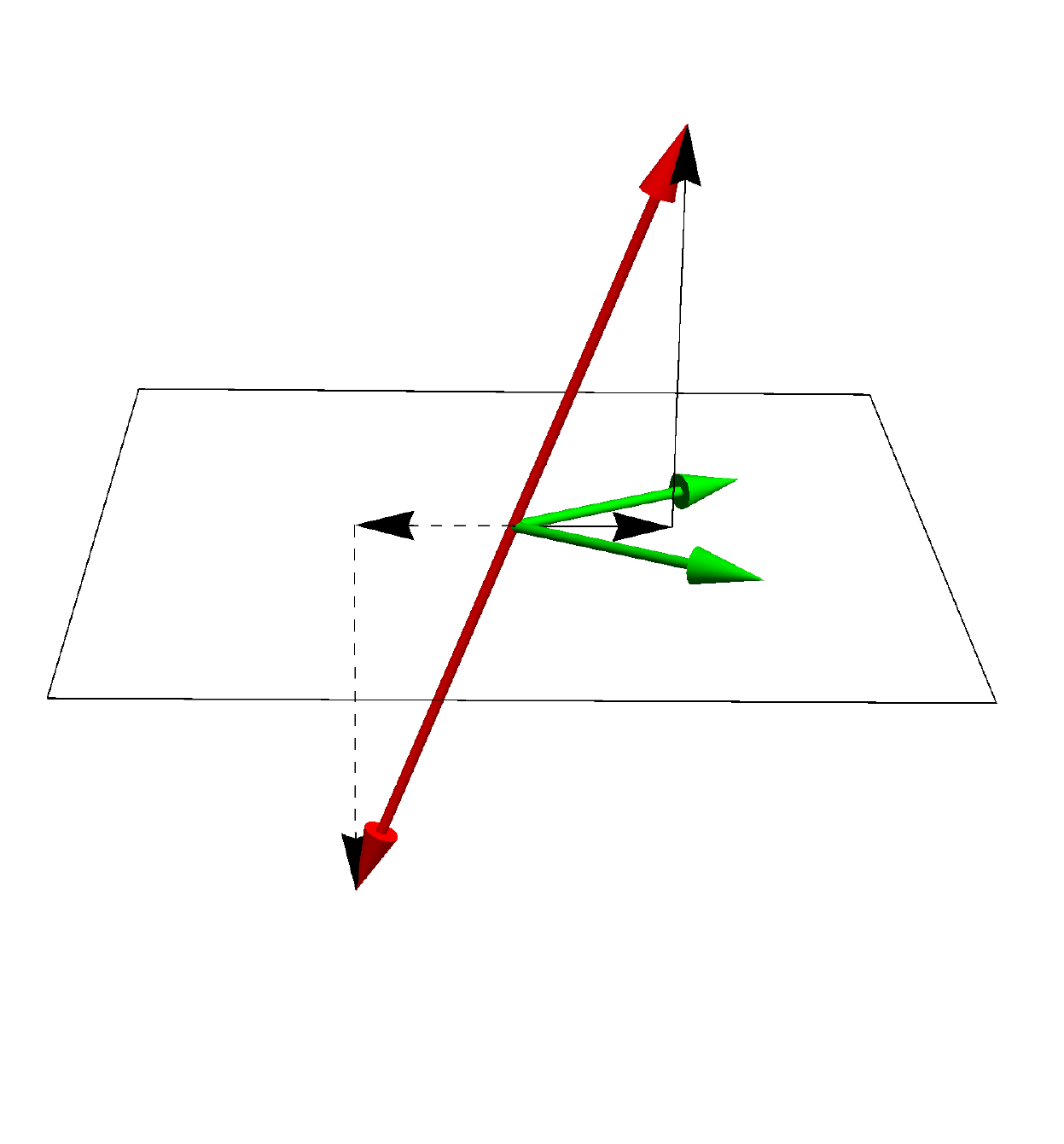}
  \caption{The (partial) construction of an extended $\bigvee$-system}
  \label{fig:fig2}
\end{figure}

\noindent Repeating the construction gives an extended configuration. For $A_2$ this is shown in Figure \ref{fig:fig3}.

\begin{figure}[h!]
  \includegraphics[width=8cm]{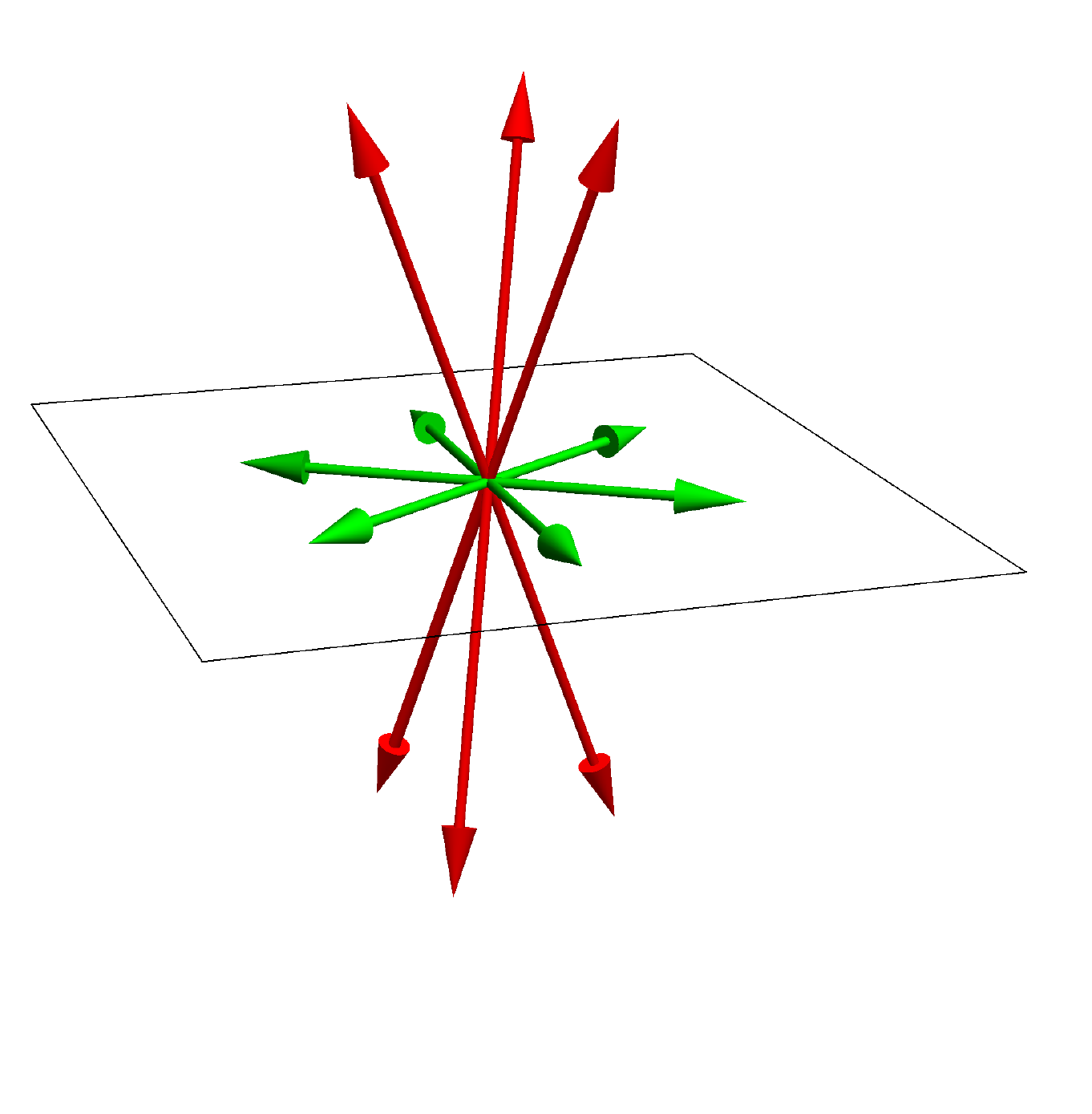}
  \caption{An extended $\bigvee$-system}
  \label{fig:fig3}
\end{figure}

\medskip

\noindent As in the previous example, the metric in the 3-dimensional space, and in particular its signature, depends on the geometry and multiplicities of these new extended vectors.
\end{example}

\noindent This extended configuration is not, at present, a $\bigvee$-system. One now imposes the $\bigvee$-conditions to obtain algebraic conditions which restrict the new data. It is at this stage that the small-orbit condition comes directly into play - it enables one to understand the 2-dimensional configurations on each two-plane $\Pi\,.$ These algebraic conditions are given in Lemma \ref{metriclemma} and Theorem \ref{basictheorem}.

This construction has, so far, not resulted in any new $\bigvee$-systems. But the results shed a new light on the geometry of the configurations and their origins which enables the link with extended-affine Weyl groups, their Frobenius and almost-dual Frobenius manifolds, and trigonometric $\bigvee$-systems, to be made.

\section{Extended $\bigvee$-systems}\label{section2}

\subsection{Extended configurations}

We begin by extending $V$ by a 1-dimensional space $V^\perp\,,$
\[
V^{ext} = V \oplus V^\perp
\]
where $V^\perp={\rm span}\{n^\vee\}$ and $V$ and $V^\perp$ are perpendicular subspaces of $V^{ext}\,.$ The metric $(\,,)^{ext}$
on $V^{ext}$ is determined by the original metric on $V$ and the value $(n^\vee,n^\vee)^{ext}$ which defines the perpendicular scale. Thus
\begin{equation}
(z_o+z^\perp,z_o+z^\perp)^{ext} = (z_o,z_o) + (z^\perp,z^\perp)^{ext}\,, \qquad z_o \in V\,,z^\perp \in V^\perp\,.
\label{extendedmetric}
\end{equation}
With this one can define the covector $n\,.$ Once this extended space has been defined one can define the extended configuration $\mathcal{U}^{ext}\,.$

\begin{defn} Let $\mathcal{U}$ be a $\bigvee$-system with an invariant small orbit $\vartheta_s\,.$ The extended configurations $\mathcal{U}^{ext}$ are defined as:
\[
\mathcal{U}^{ext} = \mathcal{U} \cup \{ \pm (w+n)\,,w \in \vartheta_s\} \cup \{ \pm n \}\,.
\]
We exclude reducible configurations, i.e. trivial extensions of the type $\mathcal{U} \cup  \{ \pm n \}$.
The corresponding multiplicities for the new covectors will be denoted $h_w$ (corresponding to the new covectors $\pm(w+n)$ and $h_n$ (corresponding to the new covectors $\pm n$).
\end{defn}

We now have two metrics on $V^{ext}$, the canonical metric given by (\ref{canonicalmetric}) (now summed over the extended configuration) and the orthogonal decomposition given by (\ref{extendedmetric}). The following Lemma gives necessary and sufficient conditions for these to be equal.

\begin{lem} \label{metriclemma}

Let $\mathcal{U}$ be a $\bigvee$-system with an invariant small orbit $\vartheta_s\,.$
The two metrics agree, i.e.
\[
h_{\mathcal{U}^{ext}} (x,y)^{ext} = \sum_{\alpha \in \mathcal{U}^{ext}} h_\alpha \, \alpha(x) \alpha(y)\,, \qquad x\,,y \in V^{ext}\,,
\]
if and only if
\[
h_{\mathcal{U}} + 2 h_{\vartheta_s} = 2 \{ h_n + \sum_{w\in \vartheta_s } h_w \} (n^\vee,n^\vee)\,.
\]
With this, $h_{\mathcal{U}^{ext}} = h_\mathcal{U} + 2 h_s\,.$
\end{lem}

\begin{proof} We prove this in the case $x=y\,,$ the full result then follows from polarization. Decomposing $z=z_o+z^\perp$ gives
\begin{eqnarray*}
\sum_{\alpha \in \mathcal{U}^{ext}} h_\alpha \alpha(z)^2 & = & \sum_{\alpha \in \mathcal{U}} h_\alpha \alpha(z)^2 + \sum_{w\in \vartheta_s}\left( h_w \left[ (w+n)(z) \right]^2 + h_w\left[-(w+n)(z)\right]^2 \right)  \\ &&+h_n\left[n(z)\right]^2 + h_n \left[-n(z)\right]^2\,,\\
&=& \sum_{\alpha \in \mathcal{U}} h_\alpha \alpha(z_o)^2 + 2 \sum_{w\in\vartheta_s} h_w \left[ w(z_o) + n(z^\perp)\right]^2 + 2 h_n n(z^\perp)^2\,,\\
&=& \left(h_\mathcal{U} + 2 h_s \right) (z_o,z_o) + ( 2 h_n + \sum_{w\in\vartheta_s} h_w ) n(z^\perp)^2\,
\end{eqnarray*}
where the invariant conditions have been used to derive the last line.

Since ${\rm dim} V^\perp =1\,, z^\perp= \mu n^\vee$ for some scalar $\mu\,,$ so
\[
(n^\vee,z^\perp)^2 = (n^\vee,n^\vee) (z^\perp,z^\perp)\,.
\]
Thus
\[
\sum_{\alpha \in \mathcal{U}^{ext}} h_\alpha \alpha(z)^2 = \left(h_\mathcal{U} + 2 h_s \right) (z_o,z_o) + ( 2 h_n + \sum_{w\in\vartheta_s} h_w ) (n^\vee,n^\vee) (z^\perp,z^\perp)\,.
\]
Since $(z,z)=(z_o,z_o) + (z^\perp,z^\perp)$ the result follows.
\end{proof}

\begin{example}\label{AnEx}
Let $\mathcal{U}=\mathcal{R}_{A_n}\,.$ We assume for now (these conditions will follow from the imposition of the $\bigvee$-conditions) that $h_n=0\,,h_\alpha=1$ for $\alpha\in \mathcal{R}_{A_n}$ (this fixes
$h_\mathcal{U} = 2(n+1)\,,$ the (dual) Coxeter number of $A_n$), and $h_w={\rm constant}$ for $w \in\vartheta_s\,.$ There are two (separate) small orbits - the orbit (under the $A_n$-group) of the highest and lowest weight vectors (the two families reflecting the symmetry of the Coxeter/Dynkin diagram). We first find the constant $h_s\,.$ From symmetry/invariant theory it follows that
\begin{eqnarray*}
\sum_{w \in \vartheta_s} h_w w(z_o)^2 &=& h_s (z_o,z_o)\,,\\
\sum_{w\in \vartheta_s} h_w w(z_o) &=&0
\end{eqnarray*}
for $z_o\in V\,.$ Thus the invariant conditions are automatically satisfied. To find $h_s$ we let $z_o=\alpha$ and sum over $\alpha\in \mathcal{R}_{A_n}\,.$ Thus
\[
h_w \sum_{w\in\vartheta_s} \sum_{\alpha \in \mathcal{R}_{A_n}} (w,\alpha)^2 = h_w \sum_{\mathcal{R}_{A_n}} (\alpha,\alpha)\,.
\]
Since $(\alpha,\alpha)=2$ and $\sum_{\alpha \in \mathcal{R}_{A_n}} (w,\alpha)^2 = 2(n+1) (w,w)\,,$ together with $\# \mathcal{R}_{A_n}=n(n+1)$, $\#\vartheta_s=n+1$ and $(w,w)=n/(n+1)\,,$ it follows that $h_s=h_w\,.$ The number of
elements in $\vartheta_s$ follows from Serganova's classification and standard properties of weight vectors.

Having found $h_s$ the result of Lemma \ref{metriclemma} yields the condition
\[
(n^\vee,n^\vee) = \frac{1}{h_w} + \frac{1}{1+n}\,.
\]
Thus the construction gives a 1-parameter family of configurations, controlled by $h_w$ (or alternatively, controlled by the perpendicular scale - the length
$(n^\vee,n^\vee)$).

Example \ref{basicexample} falls into this class, with $n=1, h_w=-1\,$ This gives $(n^\vee,n^\vee)=-\frac{1}{2}$ reflecting the split signature of the metric.

\end{example}

\subsection{Imposition of the $\bigvee$-conditions}

To impose the $\bigvee$-conditions on the extended configuration $\mathcal{U}^{ext}$ it is first necessary to classify the 2-dimensional arrangements. Since $\mathcal{U}$ is a priori a $\bigvee$-system one only has to understand plane arrangements through the origin and including some combination of the vectors $\pm(w+n)\,, \pm n\,.$ It is here that the small orbit condition comes into play. Given $w_i\,,w_j \in \vartheta_s\,,$ consider the plane containing the vectors $w_i+n\,,w_j+n\,.$ Since
\[
w_i-w_j \in {\rm span}\{w_i+n\,,w_j+n \}
\]
it follows from the small orbit condition that the plane also contains an element $\alpha\in\mathcal{U}\,.$ The set of planes for which one needs to impose the $\bigvee$ conditions depend on the geometric properties of the small orbit.

As is apparent from Figure \ref{fig:A2}, unlike roots, the negative of a small orbit vector may, or may not, be a small orbit vector. If it is not, then the pure normal vectors in the extended congiguration must be absent.

\begin{lem} Let $w \in \vartheta_s\,$  and suppose $-w \not\in \vartheta_s\,.$ Then
\[
\mathcal{U}^{ext} = \mathcal{U} \cup \{ \pm (w+n)\,,w \in \vartheta_s\}\,.
\]
\end{lem}

\begin{proof} Consider the intersection of the planes containing $\{\pm n\}$ with $\mathcal{U}\,.$ Such configurations take the form

\begin{figure}[h!]
  \includegraphics[width=12cm]{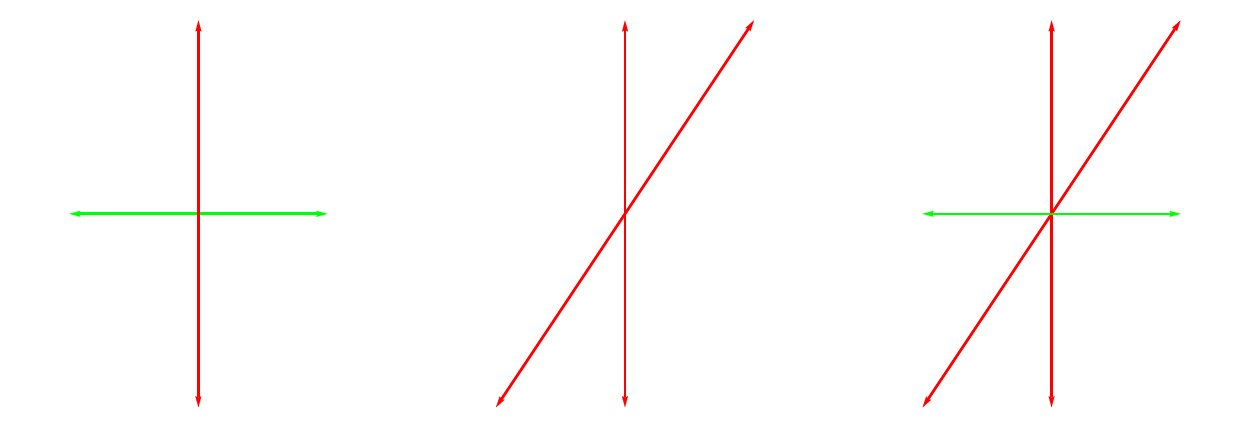}
  \caption{$2$-plane configuration including the normal direction}
  \label{A2}
\end{figure}

\noindent where the vectors $\pm (w+n)$ and/or $\pm n$ may, or may not, be present in the configuration:

\begin{itemize}

\item[$\bullet$] If $\pm(w+n)$ are not present the remaining vectors are perpendicular and the $\bigvee$-conditions are vacuous.

\item[$\bullet$] If $\pm(w+n)$ are present (so $h_w \neq 0$) the vectors $\pm \alpha$ may or may not be present. In either case the $\bigvee$-conditions imply that $h_w=0$ and so such configurations cannot occur (it is here that the assumption $-w\not\in\vartheta_s$ is used. Without this, additional terms could appear).

\end{itemize}

\noindent Thus one arrives at a reducible configuration so by definition of $\mathcal{U}^{ext}$ (which excludes such reducible configurations),
\[
\mathcal{U}^{ext} = \mathcal{U} \cup \{ \pm (w+n)\,,w \in \vartheta_s\}\,.
\]
\end{proof}

\noindent Since the small orbit vectors of $B_n$ have the property that $\pm w \in\vartheta_s$ and the small orbits vectors of $A_{n\geq 2}$ do not (as proved in \cite{S}) we define the following:

\begin{defn} For all $w\in\vartheta_s:$
\begin{itemize}
\item[(a)] if $-w \not\in\vartheta_s$ then $\mathcal{U}^{ext}$ is said to be of $A$-type;
\item[(b)] if $-w \in\vartheta_s$ then $\mathcal{U}^{ext}$ is said to be of $B$-type.
\end{itemize}
\end{defn}

\bigskip

\begin{example}
{~~~~~~~~~~~~~~~~~~~~~~~~~~~~~~~~~~~~~~~~~~~~~~~~~~~~~~~~~~~~~~~~~~~~~~~~~~~~~~~~~~}
\begin{itemize}
\item[$\bullet$] $\left(\mathcal{R}_{A_n}\right)^{ext}$ is of $A$-type;
\item[$\bullet$] $\left(\mathcal{R}_{B_n}\right)^{ext}$ is of $B$-type;
\end{itemize}
In fact, with specific choices of normalizations;
\begin{eqnarray*}
\left(\mathcal{R}_{A_n}\right)^{ext} &\cong & \mathcal{R}_{A_{n+1}}\,,\\
\left(\mathcal{R}_{B_n}\right)^{ext} &\cong & \mathcal{R}_{B_{n+1}}\,,
\end{eqnarray*}
\end{example}

\noindent One now has to impose the $\bigvee$ conditions on the planes ${\rm span}\{w_i+n\,,w_j+n\} $ and in particular on vectors in the intersection
${\rm span}\{w_i+n\,,w_j+n\} \cap \mathcal{U}\,.$ In general one can not say much about this intersection. It is at this stage that the small-orbit property comes into play\,: it enables one to know what vectors are in this set. One obtains the 2-plane configurations (see Figure \ref{ref:2planes}) and in the following Theorem the $\bigvee$-conditions are applied to these 2-plane configurations.

\begin{figure}[h!]
  \includegraphics[width=12cm]{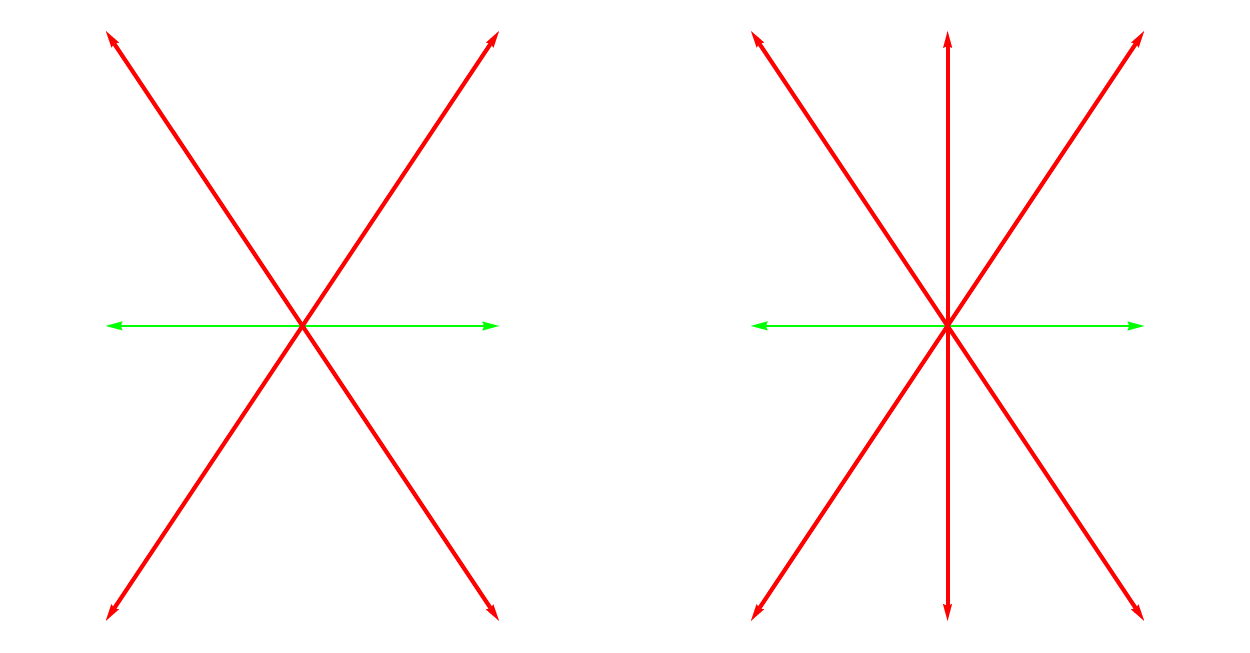}
  \caption{2-plane configurations}
  \label{ref:2planes}
\end{figure}

\begin{thm}\label{basictheorem}
The imposition of the $\bigvee$-conditions on the extended configurations $\mathcal{U}^{ext}$ results in the constraints on the data
$\{h_{w_i}\,,h_n\,,{n^\vee}\}\,:$

\begin{eqnarray*}
h_{w_i} \left[ (w_i^\vee,w_i^\vee) - (w_i^\vee,w_j^\vee) \right] & = &  h_{w_j} \left[ (w_j^\vee,w_j^\vee) - (w_i^\vee,w_j^\vee) \right]\,,\\
h_{w_i} \left[ (w_i^\vee,w_j^\vee) + (n^\vee,n^\vee) \right] & = & h_\alpha  \left[ (w_j^\vee,w_j^\vee) - (w_i^\vee,w_j^\vee) \right]\,,\\
h_{w_j} \left[ (w_i^\vee,w_j^\vee) + (n^\vee,n^\vee) \right] & = & h_\alpha  \left[ (w_i^\vee,w_i^\vee) - (w_i^\vee,w_j^\vee) \right]\,,
\end{eqnarray*}
where $\alpha = w_j-w_i\,,$ and (if the system is of $B$-type):

\[
h_n (n^\vee,n^\vee) = h_\alpha (w^\vee,w^\vee) + 2 h_w \left\{ (w^\vee,w^\vee) - (n^\vee,n^\vee)\right\}\,.
\]

\end{thm}

\noindent The proof is by direct calculation and details will not be given. It follows from applying Definition \ref{Vdef} to all 2-plane configurations containing extended vectors: details may be found in \cite{thesis}. It is here that the small orbit geometry comes into play - it gives control over the vectors in $\Pi \cap \mathcal{U}^{ext}$.  Note that this over-constrains the data $\{h_w\,,{n^\vee}\}$ but the constraints are completely determined by the geometry of small orbits.

\begin{example} Let $\mathcal{U}=\mathcal{R}_{A_n}\,.$  Then from standard properties of weight vectors,
\[
(w_i^\vee,w_j^\vee) = \delta_{ij} - \frac{1}{n+1}\,.
\]
With these one obtains $h_{w_i}=h_{w_j}\,,$ i.e. all the $h_{w_i}$ are constants, and, with the normalization $h_\alpha=1\,,\alpha\in\mathcal{U}\,,$ the condition
\[
(n^\vee,n^\vee) = \frac{1}{h_w} + \frac{1}{1+n}\,.
\]
Note, this is the same condition that comes from Example \ref{AnEx}. Thus imposition of the $\bigvee$-conditions yields no new constraints.
\end{example}

\begin{example}

Let $\mathcal{U} = \mathcal{R}_{G_2}\,.$ We normalise these roots - generated by simple roots $\alpha$ and $\beta$ by the conditions $(\alpha,\alpha)=2\,,(\alpha,\beta)=-3\,,(\beta,\beta)=6\,.$ From (\ref{canonicalmetric}) one finds
\[
6 h_s + 18 h_l = h_\mathcal{U}\,,
\]
where $h_s$ and $h_l$ are the multiplicities of the short and long roots. From \cite{S} the small orbit is the $A_2$-subsystem generated by the set
$\{\pm\alpha\,,\pm(\alpha+\beta)\,,\pm(2\alpha+\beta)\}$ and using this one finds
\[
h_{\vartheta_s} = 6 h_w\,.
\]
Lemma \ref{metriclemma} implies
\[
3 h_s + 9 h_l + 6 h_w = \left( h_n + 6 h_w \right) \, (n^\vee,n^\vee)\,.
\]
Theorem \ref{basictheorem} then implies the equation
\begin{eqnarray*}
 - 3 h_l + h_n \left\{-1 + (n^\vee,n^\vee)\right\} & = & 0 \,,\\
 -2 h_s - 4 h_w + \left\{ h_n + 2 h_w \right\} (n^\vee,n^\vee) & = & 0\,.
 \end{eqnarray*}
 Assuming $h_w\neq 0$ (otherwise the construction collapses) one can solve equations to obtain the extended configuration data in terms of the original $\bigvee$-data $\{h_s\,,h_l\}\,:$
 \begin{eqnarray*}
 h_w & = & \frac{1}{2} (h_s - 3 h_l)\,,\\
 h_n & = & \frac{3(h_s - 3 h_l)^2}{h_s + 3 h_l}\,,\\
 (n^\vee,n^\vee) & = & \frac{h_s + 3 h_l}{h_s - 3 h_l}\,.
 \end{eqnarray*}
 Note that one requires a slight constraint on the original data: $h_s \neq \pm 3 h_l\,.$ This extended $\bigvee$-system coincides, after some linear algebra and redefinitions, to the system $G_3(t)$ presented in \cite{FV}.

 \end{example}

\section{Generalized Root systems and their extensions}

The examples above have all been based on the system $\mathcal{U}$ being the root system of a (finite) Coxeter group. A wider class a examples come from
generalized root systems. Following \cite{S} and \cite{FV} :

\begin{defn}
Let $V$be a finite-dimensional complex vector space with a non-degenerate bilinear form $(\,,\,)\,.$ The finite set $\mathcal{U} \subset V \setminus \{0\}$ is called a generalized root system if the following conditions are fulfilled:
\begin{enumerate}
\item $\mathcal{U}$ spans $V$ and $\mathcal{U}=-\mathcal{U}$\,; \\
\item if $\alpha\,,\beta \in \mathcal{U}$ and $(\alpha,\alpha)\neq 0$ then $2 \frac{(\alpha,\beta)}{(\alpha,\alpha)} \in \mathbb{Z}$ and
$r_\alpha(\beta) = \beta - 2 \frac{(\alpha,\beta)}{(\alpha,\alpha) }\alpha \in \mathcal{U}\,;$\\
\item if $\alpha \in \mathcal{U}$ and $(\alpha,\alpha)=0$ then for any $\beta\in\mathcal{U}$ such that $(\alpha,\beta)\neq 0$ at least one of the vectors
$\beta+\alpha$ or $\beta-\alpha$ belongs to $\mathcal{U}$\,.\\
\end{enumerate}

\end{defn}
The classification of irreducible generalised root systems was given in \cite{S}. In the following subsections we analyze certain classes of such generalized root
systems, constructing small-orbits and hence extended configurations on which the $\bigvee$-conditions are imposed. Such configurations have much in common with the admissible deformations of generalized root systems \cite{SV}. We also give, for the two classical series $A(m,n)$ and $B(m,n)$, an interpretation of a \lq reflection\rq~in an isotropic root and a model for such an action. Roughly one has two Coxeter groups which interchange  - depending on the sign of $(\alpha,\alpha)$ - the zeros and poles of a rational function, and interchange a zero and a pole in the case of an isotropic root. Such an interpretation originates from the superpotential construction for almost-dual Frobenius manifolds.

\subsection{The series $A(m,n)$}

Consider the following data:
\begin{eqnarray*}
\mathcal{U} & = & \left\{ \alpha_{ij} := e_i-e_j\,, i\neq j\,,\quad i,j=0\,,\ldots\,,n+m\right\}\,,\\
\varepsilon_i & = & \begin{cases}  +1 &\mbox{if } i=0\,,\ldots\,,m\,,\\ -1 & \mbox{if } i=m+1\,,\ldots\,, n+m \end{cases}\,,\\
g& = & \left.\sum_{i=0}^{m+n} \varepsilon_i (dz^i)^2\right|_{\sum \varepsilon_j z^j=0}\,.
\end{eqnarray*}
This constitutes a generalized root system (recall, elements of $\mathcal{U}$ are covectors and the metric $g$ defines the vectors) and, as may be
easily verified, a $\bigvee$-system with multiplicities $h_{\alpha_{ij}} = \varepsilon_i \varepsilon_j\,.$ Note
\[
(\alpha_{ij}^\vee, \alpha_{ij}^\vee ) = \varepsilon_i+\varepsilon_j\,,
\]
so the squared lengths can be $+2\,,0$ or $-2\,,$ reflecting the split signature of the metric $g\,.$ If $n=0$ the system reduces to the standard Coxeter root system for $A_m\,.$

The small orbit is defined as follows \cite{S}:
\[
\vartheta = \left\{ w_i:=e_i + \frac{1}{n-{m+1}} \sum_{r=0}^{m+n} \varepsilon_r e_r \right\}
\]
(so trivially, $\alpha_{ij} = w_i-w_j$). That this is the only small orbit (up to sign) was proved in \cite{S}.

Applying Theorem \ref{basictheorem} gives the multiplicities $h_{w_i} = \varepsilon_i h_w$ for some $h_w\,,$ and with this the invariant conditions imply that $h_\vartheta=h_w\,.$ Finally one obtains the
condition
\[
\frac{1}{h_w} + \frac{1}{m+1-n} = ( n^\vee,n^\vee )\,.
\]
This then gives an extended $\bigvee$-system $\mathcal{U}^{ext}\,,$ with one free parameter: $h_w$ fixes the perpendicular scale, or visa-versa.

\subsection{The series $B(m,n)$}
Consider the following data:
\begin{eqnarray*}
\mathcal{U} & = & \left\{ \pm e_i \pm e_j\,, i\neq j\,, i,j=1\,,\ldots\,,n+m\right\}\cup
\left\{ e_i\,,i=1\,,\ldots\,,n+m\right\}\,,\\
\varepsilon_i & = & \begin{cases}  +1 &\mbox{if } i=1\,,\ldots\,,m\,,\\ -1 & \mbox{if } i=m+1\,,\ldots\,, n+m \end{cases}\,,\\
g& = & \sum_{i=1}^{m+n} \varepsilon_i (dz^i)^2\,.
\end{eqnarray*}
This constitutes a generalized root system (recall, elements of $\mathcal{U}$ are covectors and the metric $g$ defines the vectors) and, as may be
easily verified, a $\bigvee$-system with multiplicities $h_{\pm e_i\pm e_j} = h \varepsilon_i \varepsilon_j$ and $h_{\pm e_i} = 2 \varepsilon_i(2 \varepsilon_i+\gamma)$ with $\gamma$ arbitrary. Note

\[
( \alpha_{ij}^\vee, \alpha_{ij}^\vee ) = \varepsilon_i+\varepsilon_j\,,\quad ( \alpha_i^\vee,\alpha_i^\vee ) = \varepsilon_i\,,
\]
(where $\alpha_{ij}=\pm e_i\pm e_j$ and $\alpha_i=\pm e_i$) so the squared lengths can be $+2\,,+1\,,0\,,-1$ or $-2\,,$ reflecting the split signature of the metric $g\,.$ If $n=0$ the system reduces to the standard Coxeter root system for $B_m\,.$

The small orbit is defined as follows \cite{S}
\[
\vartheta = \left\{\pm e_i \right\}
\]
That this is the only small orbit (up to sign) was proved in \cite{S}.

Applying Lemma \ref{metriclemma} and Theorem \ref{basictheorem} gives the multiplicities $h_{w_i} = \varepsilon_i h_w$ and the conditions
\begin{eqnarray*}
h_w & = & \frac{h}{(n^\vee,n^\vee)}\,,\\
h_n & = & \frac{2h}{(n^\vee,n^\vee)} + \frac{h \gamma}{(n^\vee,n^\vee)}\,.
\end{eqnarray*}
This then gives an extended $\bigvee$-system $\mathcal{U}^{ext}\,,$ with one free parameter: the perpendicular scale.

\subsection{Symmetries of the extended configurations}

Consider again the extended configuration shown in Figure \ref{fig:fig3}. The original configuration $\mathcal{U} = R_{A_2}$ is by definition invariant under reflections - the group $A_2$ - as are the small orbit vectors (which form an irregular orbit). The reflection generated by these roots can be extended to the whole space (acting trivially in the perpendicular direction). Thus the whole 3-dimensional configuration is invariant under $A_2$ but not, in general, under reflections generated by the new roots $\pm(w+n)\,.$ Clearly this idea generalises - any symmetry is inherited by the extended configuration.

However, we now have to consider the generalized root system of $A(m,n)$ which has roots of positive, negative and zero length.
The non-isotropic roots define reflections in the hyperplanes $\alpha_{ij}({z})=0$ and the whole configuration $\mathcal{U}^{ext}$ is invariant under reflections in the original roots.

More explicitly, if $\alpha_{ij}$ is a non-isotropic root then it is easy to show, using the small-orbit property, that
\[
r_{\alpha_{ij}} (w_k) = \begin{cases} w_k & \mbox{if } {\sl i,j,k}, \mbox{ distinct},\\
w_i &\mbox{if } j=k\,,\\
w_j &\mbox{if } i=k\,,
\end{cases}
\]
(on writing $\alpha_{rs}=w_r-w_s$) and the reflection $r_{\alpha_{ij}}(\alpha_{rs})$ can then be found using linearity. However, as an abstract group, one may define a transformation
$r_{\alpha_{ij}}$ for all roots (including isotropic roots) by the above formulae. While $\left(r_{\alpha_{ij}}\right)^2=id\,,$ the lengths of the covectors are not preserved under the action of $r_{\alpha_{ij}}\,.$
For example, since
\[
(w_i^\vee,w_j^\vee) = \varepsilon_i \delta_{ij} + \frac{1}{n-(m+1)}
\]
a \lq reflection\rq~in an isotropic root will change the length of the roots. With this interpretation/definition of $r_{\alpha_{ij}}$ the whole extended
configuration $\mathcal{U}^{ext}$ is invariant under the above
action of $r_{\alpha_{ij}}\,.$

If one extends the action of $r_{\alpha_{ij}}$ to the multiplicities $\varepsilon_i$ via the formula
\[
r_{\alpha_{ij}} (\varepsilon_k) = \begin{cases} \varepsilon_k & \mbox{if } {\sl i,j,k}, \mbox{ distinct},\\
\varepsilon_i &\mbox{if } j=k\,,\\
\varepsilon_j &\mbox{if } i=k\,,
\end{cases}
\]
then the plane $\sum \varepsilon_i z^i=0$ is invariant under the action of $r_{\alpha_{ij}}$ (which now interchanges {\sl both} $z^i$ with $z^j$ and $\varepsilon_i$ with $\varepsilon_j$).

A model for this comes from the space of rational functions
\[
\lambda\left( \{z^k\},\{\varepsilon_k\}\,: z \right) = \left.\prod_{k=0}^{m+n} (z-z^k)^{\varepsilon_k}\right|_{\sum \varepsilon_j z^j=0}\,.
\]
So for all roots $\alpha_{ij}\,,$
\[
\lambda\left( \{r_{\alpha_{ij}} z^i\},\{r_{\alpha_{ij}}\varepsilon_i\}\,: z \right)=
\lambda\left( \{z^i\},\{\varepsilon_i\}\,: z \right)\,.
\]
If $n=0$ then one recovers the standard invariant polynomial model for $A_m$
\[
\lambda(z) = \left.\prod_{i=0}^m (z-z^i) \right|_{\sum z^r=0}
\]
which is invariant under the interchange of zeros. Repeating the argument for $B(m,n)$ results in the rational function
\[
\lambda\left( \{z^k\},\{\varepsilon_k\}\,: z \right) = \prod_{k=1}^{m+n} \left(z^2-(z^k)^2\right)^{\varepsilon_k}\,.
\]
If $n=0$ then one recovers the standard invariant polynomial model for $B_m\,.$

\section{Legendre transformations and Trigonometric $\bigvee$-systems}\label{Legendre}

A symmetry of the WDVV equations are transformations
\begin{eqnarray*}
z^\alpha & \mapsto & {\hat{z}^\alpha} \,,\\
g_{\alpha\beta} & \mapsto & {\hat{g}}_{\alpha\beta}\,,\\
F & \mapsto & {\hat{F}}
\end{eqnarray*}
that preserves the equations. Here the $\{z^\alpha\}$ are the flat coordinates of the metric $g\,.$ In \cite{D1} Legendre-type transformations were constructed. These take the form of the transformations
\begin{eqnarray*}
\hat{z}_\alpha & = & \partial_{z^\alpha} \partial_{z^\kappa} F\,,\\
\frac{\partial^2 \hat{F}}{\partial \hat{z}^\alpha \hat{z}^\beta } & = & \frac{\partial^2 F}{\partial z^\alpha \partial z^\beta}\,,\\
g_{\alpha\beta} & =& {\hat{g}}_{\alpha\beta}\,.
\end{eqnarray*}
Such transformations are generated by a constant vector field $\partial = \frac{\partial~}{\partial z^\kappa}\,,$ where $\kappa=1,\ldots,n\,.$

In this section we apply such a Legendre transformation to the solution of the WDVV equations given by extended $\bigvee$-systems, as constructed above. Such systems have a distinguished vector that may be used to define the Legendre transformation, namely the vector in the perpendicular, or extended, direction
$\partial = \frac{\partial~}{\partial z^\perp}\,.$

Since $\mathcal{U}^{ext} = \mathcal{U} \cup \mathcal{U}^\prime\,,$ where $\mathcal{U}$ is the original, unextended $\bigvee$-system, the new variables are given purely in terms of the vector in $\mathcal{U}^\prime$ and hence only on the small-orbit data. In particular:
\begin{equation}
\hat{z}_\alpha = \frac{\partial^2~}{\partial z^\alpha \partial z^\perp}
\left\{
\sum_{\beta \in \mathcal{U}^\prime} h_\beta \beta(z)^2 \log \beta(z)
\right\}
\label{changeofvariable}
\end{equation}
since terms involving $\alpha(z)$ with $\alpha\in\mathcal{U}$ are $z^\perp$ independent. The difficulty in applying such a transformation is computational: one has to invert the above change of variables. In the following example we invert these equations for the case when $\mathcal{U} = \mathcal{R}_{A_n}\,.$ The procedure is very general and may easily be applied to other systems.

\begin{example}
We study the case when $\mathcal{U} = \mathcal{R}_{A_n}\,,$ when the the original system is the set of roots of the $A_n$ Coxeter group.
To do this we utilise the fact that for $A_n$ we have $\# \vartheta_s=n+1$ so we can use $n$ of them as a a basis for $V\,.$ We label these
covectors $w_i\,,i=0\,,\ldots\,,n$ with $w_0=-\sum_{i=1}^n w_i\,.$

Using this basis we define (recall $z=z_o + z^\perp$)
\begin{eqnarray*}
z_i & = & w_i(z)\,, \quad i=0\,,\ldots\,, n\\
z_\perp & = & n(z)\,.
\end{eqnarray*}
Note that $\sum_{i=0}^n z_i = 0\,.$ With this, the change of variables given by (\ref{changeofvariable}) reduces to
\begin{eqnarray*}
{\hat z}_i & = & \frac{\partial^2~}{\partial z_i \partial z_\perp} \left\{ \sum_{r=0}^n 2 h_\beta (z_i+z_\perp)^2 \log(z_i+z_\perp)\right\}\,,\\
{\hat z}_\perp & = & \frac{\partial^2~}{\partial (z_\perp)^2} \left\{ \sum_{r=0}^n 2 h_\beta (z_i+z_\perp)^2 \log(z_i+z_\perp)\right\}\,.
\end{eqnarray*}
On absorbing constants, or, using the quadratic freedom in the definition of $F$ one obtains the simple system
\begin{eqnarray*}
{\hat z}_i & = & 4 h_\beta \left\{  \log(z_i+z_\perp) - \log(z_0+z_\perp) \right\}\,,\\
{\hat z}_\perp & = & 4 h_\beta \sum_{r=0}^n \log(z_i+z_\perp)
\end{eqnarray*}
which is straightforward to invert. This yields
\begin{eqnarray*}
z_i+z_\perp & = & \displaystyle{e^{\frac{1}{4 h_\beta(n+1)} ({\hat{z}}_0+{\hat{z}}_\perp)} \,. \,e^{\frac{1}{4h_\beta} {\hat z}_i} }\,,\qquad i=1\,,\ldots\,,n\,,\\
z_0+z_\perp & = & \displaystyle{e^{\frac{1}{4 h_\beta(n+1)} ({\hat{z}}_0+{\hat{z}}_\perp)}\,.}
\end{eqnarray*}
One could go further and solve these, but this is not actually required. All that is required are the terms $\alpha(z)$ and $\beta(z)\,.$ In fact, the above formulae are precisely the $\beta(z)$-terms, and to find the $\alpha(z)$-terms one can again use the small-orbit property and write each
$\alpha$ as the difference of two small-orbit covectors. Thus if $\alpha=w_i-w_j$ then
\begin{eqnarray*}
\alpha(z) & = & (w_i-w_j)(z)\,,\\
& = &
\begin{cases}
z_i-z_j & \mbox{if } i\,,j \neq 0 \,,\\
z_0-z_j &\mbox{if } i =0\,,j\neq 0\,.
\end{cases}\\
&=&
\begin{cases}
e^{\frac{1}{4 h_\beta(n+1)} ({\hat{z}}_0+{\hat{z}}_\perp)}  \left(e^{\frac{1}{4h_\beta} {\hat z}_i }-e^{\frac{1}{4h_\beta} {\hat z}_j} \right)\,,& \mbox{if } i\,,j \neq 0 \,,\\
&\\
e^{\frac{1}{4 h_\beta(n+1)} ({\hat{z}}_0+{\hat{z}}_\perp)} \left(1-\,e^{\frac{1}{4h_\beta} {\hat z}_i }\right)\,,& \mbox{if } i =0\,,j\neq 0\,.
\end{cases}\\
\end{eqnarray*}

To complete the Legendre transformation one has to integrate the equations to find $\hat{F}$. Using these formulae
one finds, schematically, that
\[
\frac{\partial^2 {\hat F}}{\partial {\hat z}_i \partial {\hat z}_j} = {\rm linear~term~} + \sum_\sigma {\rm terms~involving~}\left\{\log\left[ 1 - e^{\sigma({\hat z})} \right]\right\}
\]
(where the $\sigma$ are certain covectors - linear combinations of the original covectors) with similar formulae for the ${\hat z}_\perp$-derivatives. It is important to note that ${\hat z}_\perp$ only occurs in the linear terms in these expressions.
Integrating yields a solution of the WDVV equations of trigonometric type,
\[
{\hat F} = {\rm cubic~term~} + \sum_{\sigma\in{\widehat{\mathcal{U}}}} c_\sigma Li_3\left( e^{\sigma(z)} \right)
\]
where, again, ${\hat z}_\perp$ only occurs in the cubic term.

\end{example}

\noindent Explicit examples of this kind may be found in \cite{F,RS}.

\begin{example}
Applying this Legendre transformation to Example \ref{basicexample} yields the prepotential \cite{RS}
\[
{\hat F} = \frac{1}{24} z_\perp^3 - \frac{1}{8} z_\perp z^2 + \frac{1}{2} \left\{  Li_3 \left(e^z\right)  +  Li_3 \left(e^{-z}\right)  \right\}\,.
\]
This is the almost-dual prepotential associated to the prepotential
\[
F=\frac{1}{2} t_1^2 t_2 + e^{t_2}\,,
\]
which is constructed from the extended affine Weyl group $A_1^{(1)}\,.$
\end{example}

\noindent The above example shows a connection between extended $\bigvee$-systems and extended affine Weyl groups. This link - via a Legendre transformation - will form the subject of the next section.

\section{Extended $\bigvee$-systems and almost duality for extended affine Weyl orbit spaces}\label{hurwitz}

In the early 1990's two classes of solutions to the WDVV equations were constructed from a finite Coxeter group $W\,.$ On the orbit space
$\mathbb{C}^N/W$ a polynomial solution was found \cite{D1} using the Saito construction \cite{Saito}. But, motivated by the appearance of the WDVV equation in Seiberg-Witten theory, solutions of the form
\[
F^\star = \frac{1}{4} \sum_{\alpha \in R_W} \alpha(z)^2 \log \alpha(z)
\]
had also been found. Later in the decade, the link between these two classes of solutions was found, via the notion of almost-duality \cite{D2}.

Given a Frobenius manifold $F=\{\mathcal{M}\,,\circ\,,e\,,E\,,\eta\}$ (we denote the manifold and the corresponding prepotential by the symbol
$F$) one may obtain a so-called almost-dual Frobenius manifold $F^\star = \{ \mathcal{M} \backslash \mathcal{D}\,, \star\,,E\,, g\}$ by defining a new metric and multiplication by the formulae
\begin{eqnarray*}
g(X,Y) & = & \eta(E^{-1} \circ X,Y)\,, \\
X \star Y & = & E^{-1} \circ X \circ Y\,.
\end{eqnarray*}
This new metric turns out to the flat (by virtue of the Frobenius manifold axioms) and compatible with the new multiplication. Using these facts Dubrovin showed that the $\star$-multiplication provides us with a solution of the WDVV-equations written in the flat coordinates of the
new metric $g\,.$

Thus given a Frobenius manifold $F$ we have an almost-dual manifold $F^\star\,.$ But one may apply a Legendre transformation to $F$ to get a new manifold $\hat{F}$ and apply almost duality to that. This is summarised in the diagram:
\[
\begin{array}{ccc}
F&\overset{S_\kappa}{\longrightarrow}&
{\hat F}\\
\downarrow & & \downarrow \\
F^\star
&&
{\hat F}^\star
\end{array}
\]
In \cite{RS} a twisted Legendre transformation was constructed:
\[
F^\star \overset{{\hat S}_\kappa}{\longrightarrow} {\hat F}^\star\,
\]
The vector field that generates this twisted Legendre transformation is
\[
\tilde\partial = E \circ \partial\,,
\]
i.e the generating field $\partial$ twisted by multiplication by the Euler vector field.

One of the main problems in the theory of almost-dual {\sl type} solutions to the WDVV equations (this class including $\bigvee$-systems) is to indentfy those which are precisely the almost-dual prepotentials to some Frobenius manifold. In \cite{D2} a reconstruction theorem was proved, but the theorem is extremely difficult to use in practice. It relies on solving a linear Lax pair, and finding a vector field (to play the role of $e$) which acts on this solution in a specific way. In this section we bypass this reconstruction theorem and prove the following.

\begin{thm}
Let $W$ be a finite irreducible classical Coxeter group of rank $l$ and let $\widetilde{W}$ be the extended affine Weyl group of $W$ with arbitrary marked node. Then up to a Legendre transformation, the almost dual prepotentials of the classical extended affine Weyl group orbit spaces
$\mathbb{C}^{l+1}/\widetilde{W}$ are, for specific values of the free data, the extended $\bigvee$-systems of the $\bigvee$-system $R_W\,.$
\medskip

\noindent In particular:
\begin{itemize}

\item{} The almost dual prepotential corresponding to the orbit space
\[
\mathbb{C}^{l+1}/ {\widetilde{W}^{(k)}(A_l)}
\]
is the Legendre transformation, along the extended direction, of the extended $\bigvee$-system with:
\begin{itemize}
\item[(i)] original $\bigvee$-system
\begin{eqnarray*}
\mathcal{U} & =  & R_{A_l}\,,\\
h_\alpha & = & 1 \qquad \forall \alpha \in \mathcal{U}\,;
\end{eqnarray*}
\item[(ii)] data for extension (of $A$-type):
\begin{eqnarray*}
\vartheta_s & = & \{ w(\omega_1) \, | w \in W(A_l) \}\,,\\
h_w & = & -(l+1-k) \qquad \forall w\in\mathcal{U}^{ext}\,.
\end{eqnarray*}
\item[(iii)] superpotential data: The superpotential for the extended $\bigvee$-system is given by (\ref{genA}) with
\[
{\bf k}^{ext}= \{ -(l+1-k) \,,\underbrace{1\,,\ldots\,,1}_{l+1}\}\,.
\]
\end{itemize}
\item{} The almost dual prepotential corresponding to the orbit space
\[
\mathbb{C}^{l+1}/ {\widetilde{W}^{(k)}(C_l)}
\]
with flat structure defined by the constant $m$, where $0\leq m\leq l-k$,
is the Legendre transformation, along the extended direction, of the extended $\bigvee$-system with (where $s=-2(l-(k+m))\,,:$
\begin{itemize}
\item[(i)] original $\bigvee$-system
\begin{eqnarray*}
\mathcal{U} & =  & R_{B_l}\,,\\
h_{\alpha_{short}} & = & 1 + s \,\\
h_{\alpha_{long}} & = & 1
\end{eqnarray*}
for all $\alpha_{short/long} \in \mathcal{U}\,;$
\item[(ii)] data for extension (of $B$-type):
\begin{eqnarray*}
\vartheta_s & = & \{ w(\omega_1) \, | w \in W(B_l) \}\,,\\
h_w & = & -2k \qquad \forall w\in\mathcal{U}^{ext}\,,\\
h_n & = & - 2k (s+2k)\,.
\end{eqnarray*}
\item[(iii)] superpotential data: The superpotential for the extended $\bigvee$-system is given by (\ref{genB}) with
\[
{\bf k}^{ext}= \{ -2k  \,,\underbrace{1\,,\ldots\,,1}_{l}\}\,.
\]
\end{itemize}
\end{itemize}
\end{thm}

The proof will utilize a Hurwitz space construction. For extended affine Weyl groups of type $A$ this construction was given in \cite{DZ}. For types
$B\,,C\,,D$ this was given in \cite{DSZZ}.

Hurwitz spaces are moduli spaces of pairs
$(\mathcal{C},\lambda)\,,$ where $\mathcal{C}$ is a Riemann
surface of degree $g$ and $\lambda$ is a meromorphic function on
$\mathcal{C}$ of degree $N\,.$ It was shown in \cite{D1}
that such spaces may be endowed with the structure of a Frobenius
manifold. The $g=0$ case is particularly simple - meromorphic
functions from the Riemann sphere to itself are just given by
rational functions. It is into this category of Frobenius
manifolds that the examples constructed above fall.

More specifically, the Hurwitz space $H_{g,N}(k_1\,,\ldots\,,k_l)$
is the space of equivalence classes
$[\lambda:\mathcal{C}\rightarrow\mathbb{P}^1]$ of $N$-fold branched covers\footnote{Dubrovin uses the
different notation $H_{g,k_1-1\,,\ldots\,,k_l-1}\,.$}
with:

\begin{itemize}
\item $M$ simple ramification points $P_1, \dots,
P_M\in{\mathcal{L}}$ with distinct {\it finite} images $l_1,
\dots, l_M\in {\mathbb C}\subset {\mathbb P}^1$; \item the
preimage $\lambda^{-1}(\infty)$ consists of $l$ points:
$\lambda^{-1}(\infty)=\{\infty_1, \dots,\infty_l\}$, and the
ramification index of the map $p$ at the point $\infty_j$ is $k_j$
($1\leq k_j\leq N$).
\end{itemize}

\noindent  The Riemann-Hurwitz formula
implies that the dimension of this space is $M=2g+l+N-2\,.$ One has
also the equality $k_1+\dots +k_l=N$. For $g>0$ one has to introduce
a covering space, but this is unnecessary in the $g=0$ case that will be
considered here.

In this construction there is a certain ambiguity; one has to choose a so-called primary differential (also known
as a primitive form). Different choices produce different solutions to the WDVV equations, but such solutions are
related by Legendre transformation $S_\kappa\,.$ The Hurwitz data $\{\lambda,\omega\}$ from which one constructs a
solution $F_{\{\lambda,\omega\}}$ consists of the map $\lambda$ (also known as the superpotential) and a particular
primary differential $\omega\,$ \cite{D1}. Thus, again schematically, one has:
\[
F_{\{\lambda,\omega\}} \overset{S_\kappa}{\longleftrightarrow} {\hat F}_{\{\lambda,{\hat\omega}\}}
\]
(note the map $\lambda$ does not change, though it might undergo a coordinate transformation). The metrics
$<,>\,,(,)$ and multiplications $\circ\,,\star$ are determined by calculating certain residues at the critical
points of the map $\lambda\,:$

\begin{thm}\label{superpotential}

\begin{eqnarray*}
<\partial',\partial{''}>& = & - \sum \res_{d\lambda=0}
\left\{\frac{\partial'(\lambda(v)) \partial{''}\lambda(v)
}{\lambda^\prime(v)} \, \omega \right\}\,, \\
<\partial'\circ\partial{''},\partial{'''}> & = & - \sum
\res_{d\lambda=0} \left\{ \frac{\partial'\lambda(v)
\partial{''}\lambda(v) \partial{'''}\lambda(v)}{\lambda^\prime(v)}\,\omega \right\} \,, \\
(\partial',\partial{''}) & = & - \sum \res_{d\lambda=0}
\left\{\frac{\partial'\log\lambda(v) \partial{''}\log\lambda(v)}{(\log\lambda)^\prime(v)} \, \omega\ \right\}\,,\\
(\partial'\star\partial{''},\partial{'''}) & = & - \sum
\res_{d\lambda=0}  \left\{ \frac{\partial'\log\lambda(v)
\partial{''}\log\lambda(v)\partial{'''}\log\lambda(v)}{(\log\lambda)^\prime(v)}\,\omega \right\} \,.
\end{eqnarray*}
Here $\omega$ is the primary differential.

\end{thm}

\medskip

\noindent We divide the proof into the $A$ case and the $B\,,C\,,D$ cases.

\subsection{Extended Affine Weyl orbit spaces of type $A$}

In \cite{DZ} it was shown, given an extended affine group $\widetilde{W}^{(k)}(A_l)$, that the orbit space $\mathbb{C}^{l+1}/\widetilde{W}^{(k)}(A_l)$
maybe endowed with the structure of a Frobenius manifold. Additionally, it was shown that this space is isomorphic to the space of trigonometric polynomials of bi-degree $(k,l+1-k)$, namely, functions of the form
\begin{equation}
\lambda({\tilde{z}}) = e^{ik\tilde{z}} + a_1 e^{i(k-1)\tilde{z}} + \dots + a_{l+1} e^{-i(l+1-k)\tilde{z}}
\label{AEWA}
\end{equation}
with the Frobenius manifold structures being given by Theorem \ref{superpotential} with the choice of primary differential $\omega=d\tilde{z}\,.$

This space is related, via a Legendre transformation, to the Hurwitz space $\mathcal{M}_{k,l+1-k}$ of rational functions of the form
\begin{equation}
\lambda(z) = z^k + \alpha_1 z^{k-2} + \ldots + \frac{t^\circ}{z-z^\circ} + \ldots + \frac{\alpha_{l+1}}{(z-z^\circ)^{l+1-k}}
\label{rat}
\end{equation}
with primary differential $\omega=dz\,.$ The coefficient $t^\circ$ turns out, by evaluating certain residues, to be a flat coordinate
and hence a change in primary differential is given by
\[
d\tilde{z} = \left\{ \partial_{t^\circ} \lambda(z) \right\} dz
\]
so $\tilde{z} = \log(z-z^\circ)\,,$
and this induces the Legendre transformation between the two Frobenius manifolds, i.e. this change of primary differential induces a change of variable that maps (\ref{AEWA}) to (\ref{rat}).
Thus the Frobenius manifold structures on
$\mathbb{C}^{l+1}/\widetilde{W}^{(k)}(A_l)$ and $\mathcal{M}_{k,l+1-k}$ are related by a Legendre transformation.

Rewriting the rational function (\ref{rat}) in this form
\[
\lambda(z) = \left.\frac{ \prod_{i=1}^{l+1} (z-z_i) }{(z-z^\circ)^{l+1-k}}\right|_{\sum_{i=1}^{l+1} z_i - (l+1-k) z^\circ =0}
\]
thus gives a $\bigvee$-system which is of extended type, i.e. an extension of the $A_l$ $\bigvee$-system constructed above,
with the zeros of the superpotential being flat coordinates for the metric $g\,.$ This is a special case, with
\[
{\bf k}=\{ -(l+1-k)\,,\underbrace{1\,,\ldots\,,1}_{l+1}\}
\]
of the general superpotential (\ref{genA}).

At this stage one could directly calculate the almost-dual prepotential from the
trigonometric superpotential, and the result is a solution which is derived from a trigonometric $\bigvee$-system \cite{RS}. However, we will take a different approach and perform the
twisted Legendre transformation.

\[
\begin{array}{ccc}
  \mathcal{M}_{k,l+1-k} & \xleftrightarrow{\text{Legendre~transformation}}&  {\mathbb{C}^{l+1}/\widetilde{W}^{(k)}(A_l)} \\
&&\\
\Bigg\downarrow& & \Bigg\downarrow \\
&&\\
\left\{\begin{array}{c} {\rm almost~dual}\\{\rm~prepotential~of}\\{\rm rational~type} \end{array}\right\}
& \xleftrightarrow{\text{twisted~Legendre~transformation}}
& \left\{\begin{array}{c} {\rm almost~dual}\\{\rm~prepotential~of}\\{\rm trigonometric~type} \end{array}\right\} \\
\end{array}
\]

\noindent It turns out that, with the choice $\partial=\partial_{t^\circ}$ the twisted Legendre field is actually constant.

\begin{lem}
The twisted Legendre field $E \circ \frac{\partial~}{\partial t^\circ}$ is constant in the $\{z^i\}$-variables.
Moreover, the field is perpendicular to the
space
\[
TV={\rm span}\left\{ \frac{\partial~}{\partial w^i}\,,i=1\,,\ldots\,, l \right\}\,,
\]
where the variables $z^i$ and $w^i$ are related by the linear change of variables
\begin{eqnarray*}
w^i & = & z^i - \frac{1}{l+1} \sum_{j=1}^{l+1} z^j\,, \qquad j=1\,,\ldots\,, l+1\,,\\
w^\perp & = & \frac{1}{l+1-k} \sum_{j=1}^{l+1} z^j\,,
\end{eqnarray*}
together with the constraint $\sum_{i=1}^{l+1} w^i=0\,.$
\end{lem}

\begin{proof}
We calculate
\[
g\left( E \circ \frac{\partial~}{\partial t^\circ}, \frac{\partial~}{\partial z^i} \right)
\]
Using the above formulae,
\begin{eqnarray*}
g\left( E \circ \frac{\partial~}{\partial t^\circ}, \frac{\partial~}{\partial z^i} \right) & = & \eta\left( \frac{\partial~}{\partial t^\circ}, \frac{\partial~}{\partial z^i} \right)\,,\\
& = & \sum \res_{d\lambda=0}\left\{ \frac{\frac{\partial \lambda}{\partial t^\circ} \frac{\partial\lambda}{\partial z^i} }{\lambda^\prime} dz \right\}\,,\\
& = & \sum \res_{d\lambda=0}\left\{ \frac{1}{(z-z^\circ)} \frac{\partial\lambda}{\partial z^i} \frac{\lambda}{\lambda^\prime} dz \right\}\,.\\
\end{eqnarray*}
A standard residue calculation gives
\[
g\left( E \circ \frac{\partial~}{\partial t^\circ}, \frac{\partial~}{\partial z^i} \right) = -\frac{1}{l+1-k}\,.
\]
In these coordinates,
\[
g= \sum_{i=1}^{l+1} (dz^i)^2 - \frac{1}{l+1-k} \left( \sum_{j=1}^{l+1} dz^j\right)^2\,,
\]
and hence
\[
E\circ \frac{\partial~}{\partial t^\circ} = \frac{1}{k} \sum_{j=1}^{l+1} \frac{\partial~}{\partial z^j}\,.
\]
In the $w^i$ variables,
\[
E\circ \frac{\partial~}{\partial t^\circ} = - \frac{l+1}{k(l+i-k)} \frac{\partial~}{\partial w^\perp}\,.
\]
and
\[
g=\left.\sum_{i=1}^l (dw^i)^2\right|_{\sum w^i=0} - \frac{k(l+1-k)}{l+1} (dw^\perp)^2
\]
and hence the twisted Legendre field is perpendicular to the space $TV\,.$
\end{proof}

\subsection{Extended Affine Weyl orbit spaces of type $B\,,C\,,D $}

In a similar way, given an extended affine Weyl group of type $B\,,C\,,D$ there exists Frobenius manifold structures on the corresponding orbit
space. In \cite{DZ} this was constructed for a specific choice of marked node and in \cite{DSZZ} this construction was generalized to the case of an arbitrary marked node. Thus given a extended affine Weyl goup of $C$-type, $\widetilde{W}^{(k)}(C_l)$ one can construct a Frobenius manifold structure on the orbit space ${\mathbb{C}^{l+1}/\widetilde{W}^{(k)}(C_l)} \,.$ However, unlike the $A$-case, there is an additional freedom in the choice of flat structure on the orbit space, and this freedom is defined in terms of an additional integer $0 \leq m \leq l-k\,.$ Thus the Frobenius manifold structure on the orbit space - defined by the pair $(k,l)$ -  depends on the triple $(k,l,m)\,.$ This Frobenius manifold will be denoted
$\mathcal{M}_{k,m}(C_l)\,.$

This construction also covers the orbit spaces, and their Frobenius manifold structures, for the extended affine Weyl groups
$\widetilde{W}^{(k)}(B_l)$ and $\widetilde{W}^{(k)}(D_l)\,.$ The ring of invariant polynomials (freely generated by an appropriate Chevalley-type theorem) for these groups may be obtained from those constructed for the group $\widetilde{W}^{(k)}(C_l)$ by simple changes of variable, and this leads to isomorphic Frobenius manifolds. thus it suffices to study the orbit space ${\mathbb{C}^{l+1}/\widetilde{W}^{(k)}(C_l)}$ with the Frobenius manifold structure $\mathcal{M}_{k,m}(C_l)\,.$

Furthermore, it was shown that $\mathcal{M}_{k,m}(C_l)$ coincides with the Frobenius manifold structure on the space of cosine-Laurent
series of tri-degree $(2k,2m,2l)\,,$ namely functions of the form
\begin{equation}
\lambda(\tilde{z}) = \frac{1}{(\cos^2 \tilde{z} -1)^m} \sum_{j=0}^{l} a_j \cos^{2(k+m-j)}(\tilde{z})\,,
\label{AEWB}
\end{equation}
with the choice of primary differential $\omega=d\tilde{z}\,.$

This space is related, via a Legendre transformation, to a space of $\mathbb{Z}_2$-invariant rational functions of the form
\begin{equation}
\lambda(z) = z^{2m} + \alpha_{m-1} z^{2(m-1)} + \ldots + \alpha_0 + \sum_{r=1}^{l-(k+m)} \frac{\beta_r}{z^{2r}} + \
\sum_{s=1}^k \frac{\gamma_s}{\left(z^2 - (z^\circ)^2\right)^s}
\label{ratB}
\end{equation}
with primary differential $\omega=dz\,.$ We denote this space $\mathcal{M}^{\mathbb{Z}_2}_{m,l-(k+m),k}\,.$

The coefficient $t^\circ$, defined by the term
\[
\lambda(z) = \ldots + \frac{z^\circ t^\circ}{\left(z^2-(z^\circ)^2\right)} + \ldots
\]
turns out, by evaluating certain residues, to be a flat coordinate
and hence a change in primary differential is given by
\[
d\tilde{z} = \left\{ \partial_{t^\circ} \lambda(z) \right\} dz\,.
\]
so (up to an overall constant that may be ignored) $z=iz^\circ \cot\tilde{z}\,,$
and this induces the Legendre transformation between the two Frobenius manifolds, i.e. this change of primary differential induces a change of variable that maps (\ref{AEWB}) to (\ref{ratB}).
Thus the Frobenius manifold structures on $\mathcal{M}_{k,l}(C_l)$  and the $\mathbb{Z}_2$-graded Hurwitz space $\mathcal{M}^{\mathbb{Z}_2}_{m,l-(k+m),k}$ are related by a Legendre transformation.

As in the $A$-case, the twisted Legendre transformation between the corresponding almost dual manifolds turns out to a normal Legendre transformation. The extended-$\bigvee$-systems may be easily calculated by writing the superpotential in the form
\[
\lambda(z) = \frac{ \displaystyle{\prod_{i=1}^l \left(z^2 - (z^i)^2 \right) }}{ \displaystyle{z^{2(l-(k+m))} \left(z^2 - (z^\circ)^2\right)^{2k}}}\,.
\]
Unlike the $A$-case the $z_o$ variable is not constrained: the flat coordinates for the metric are $z^i\,,i=0\,,\ldots\,,l\,.$
This is a special case, with $s=-2\left(l-(k+m)\right)$ and
\[
{\bf k}=\{    -2k\,,   \underbrace{1\,,\ldots\,,1}_l    \}
\]
of the general superpotential (\ref{genB}).

\begin{lem}
The twisted Legendre field $E \circ \frac{\partial~}{\partial t^\circ}$ is constant in the $\{z^i\}$-variables. Moreover, the field is perpendicular to the
space
\[
TV={\rm span}\left\{ \frac{\partial~}{\partial z^i}\,,i=1\,,\ldots\,, l \right\}\,.
\]
\end{lem}

\begin{proof}
Since
\[
4k \frac{\partial \lambda}{\partial t^\circ} = \frac{\partial \log\lambda}{\partial z^o}
\]
it immediately follows that, for all $i=0\,,\ldots\,,l\,:$
\[
g\left( E \circ \frac{\partial~}{\partial t^\circ}, \frac{\partial~}{\partial z^i} \right) = g\left(
\frac{1}{4k} \frac{\partial~}{\partial z^\circ} ,\frac{\partial~}{\partial z^i} \right)\,.
\]
Hence
\[
E \circ \frac{\partial~}{\partial t^\circ} = \frac{1}{4k} \frac{\partial~}{\partial z^\circ}\,.
\]
In these variables
\[
g=\sum_{i=1}^l (dz^i)^2 - 2 k (dz^\circ)^2\,,
\]
and hence it follows that the twisted Legendre field is perpendicular to the space $TV\,.$

\end{proof}

\noindent We end by noting two things:
\begin{itemize}
\item[(1)] To make the connection with Section \ref{section2}
\begin{eqnarray*}
&&TV\text{~is~spanned~by~} \left\{\frac{\partial~}{\partial w^i}\, {\rm ~for~}A{\rm ~type}\right\}\,, \left\{\frac{\partial~}{\partial z^i} {\rm ~for~}B,C,D{\rm ~type}\right\}\\
&&TV^\perp \text{~is~spanned~by~} \left\{\frac{\partial~}{\partial w^\perp}\,{\rm ~for~}A{\rm ~type}\right\}\,, \left\{\frac{\partial~}{\partial z^\perp} {\rm ~for~}B,C,D{\rm ~type}\right\}
\end{eqnarray*}
So the twisted Legendre field $E\circ \frac{\partial~}{\partial t^o}$ is perpendicular to the space $V$ (the constraint in the $A$-case, $\sum_{i=1}^{l+1} w^i=0\,,$ is just the manifestation of the standard representation of the $A_{l}$ roots system as a hyperplane
in $\mathbb{R}^{l+1}\,.$
\item[(2)]
Since $\tilde{\partial}$ is constant in these (flat)-coordinates, the twisted Legendre transformation is actually a normal Legendre transformation, and such a normal Legendre transformation has already been performed in Section \ref{Legendre}
\end{itemize}
The choice of original Legendre field $\partial=\partial_{t^\circ}$ was very special - others choices would have resulted in a non-constant twisted Legendre field. It may be shown that this special properties comes from the fact that
$\mu=-\frac{1}{2}$ lies in the spectrum of the underlying Frobenius manifold \cite{SS}.

\section{Conclusions}

As remarked earlier, and as is apparent from Figure \ref{fig:fig2}, an extended $\bigvee$-system, based, say, on the root system $R_W$ of a classical Weyl group $W$ of rank $n$, is invariant under the action of $W\,.$ But on performing a Legendre transformation one obtains configurations invariant under an {\sl extended} affine Weyl group of rank $n+1\,.$ It is therefore natural to ask what is the origin of this extra symmetry that does not appear in the extended $\bigvee$-system. The answer lies in the precise nature of the Legendre transformation ${\hat z} \leftrightarrow z\,.$ For example, in the $A_n$ example the perpendicular direction $z_\perp$ is invariant under the affine translation
\[
{\hat z}_\perp \mapsto {\hat z}_\perp + 8 \pi i h_\beta (n+1)
\]
since the transformation is exponential. Thus in ${\hat z}$-space one has a symmetry that is not apparent in the original $z$-space. Together with the original action of $W$, one thus obtains an extended affine group action.

The construction in this paper is dependent on the existence of a small orbit and for exceptional Coxeter groups such orbits do not exist. However variants of the construction in this paper do exist: the small orbit condition is sufficient for the construction, but it is certainly not necessary, as the following example shows.

\begin{example} Consider the root system of the Coxeter group $F_4\,,$
\[
\mathcal{R}_{F_4} = \mathcal{R}_{long}\cup\mathcal{R}_{short}
\]
where
\begin{eqnarray*}
\mathcal{R}_{long} & = & \left\{ (\pm 1,\pm 1,0,0),{\rm ~and~permutations}\right\}\,,\\
\mathcal{R}_{short} & = & \left\{ (\pm 1,0,0,0),{\rm ~and ~permutations}\right\}\cup\left\{ (\pm\frac{1}{2},\pm\frac{1}{2},\pm\frac{1}{2},\pm\frac{1}{2}),{\rm ~and ~permutations}\right\}\,,\\
\end{eqnarray*}
These give the standard almost-dual prepotential
\[
F^\star = \frac{1}{4} \sum_{\alpha \in \mathcal{R}_{F_4}} \alpha(z)^2 \log\alpha(z)
\]
with
\[
2 h(z,z) = \sum_{\alpha \in \mathcal{R}_{F_4}} \alpha(z)^2
\]
and $h=9\,,$ the dual Coxeter number of $F_4\,.$ The long roots do {\rm not} form a small orbit, but it does have the following property: for $w_1\,,w_2\in \mathcal{R}_{long}$ with $w_1\neq\pm w_2\,,$

\begin{itemize}
\item[(i)] either $w_1+w_2 \in \mathcal{R}_{long}$ or $w_1-w_2 \in \mathcal{R}_{long}$
\end{itemize}
\noindent or
\begin{itemize}
\item[(ii)] $w_1$ and $w_2$ are perpendicular.
\end{itemize}

\noindent One can then extend the vectors in $\mathcal{R}_{long}$ to obtain the extended $\bigvee$-system
\[
\mathcal{U} = \mathcal{R}_{F_4} \cup
\{\pm(w+n)\,, w\in \mathcal{R}_{long}\} \cup \{ \pm n\}\,.
\]
with the respective multiplicities $h_\alpha = 1\,,\frac{1}{2}\,,3$ respectively. This gives an extended $F_4$ solution to the WDVV equations. By construction the 5-dimensional extended configuration is invariant under the Coxeter group $F_4\,.$ This configuration was obtained in \cite{FV2} by restricting the $E_8$ system onto a certain discriminant, in their notation this is the solution $(E_8,A_1^3)\,.$
\end{example}

\noindent Whether a similar construction may be applied to the remaining exceptional cases is under investigation. However the construction in \cite{DZ}, coupled with almost-duality, guarantees a trigonometric $\bigvee$-system for such exceptional cases (for a specific marked node in \cite{DZ}, and conjecturally for an arbitrary marked node). Whether such systems are the Legendre-transformed versions of some extended rational $\bigvee$-system is unknown, though it is natural to conjecture that they are. More generally, a natural question to ask is whether there is some direct map between rational $\bigvee$-systems and trigonometric $\bigvee$-systems, and if not, to find under what conditions it does exist.

Further examples too would be of interest. there has been recent work on the classification of $\bigvee$-systems \cite{L,ScV} and it would be interesting to see if extended version of these system exist. One could ask, for example, how the matroid for the extended systems can be constructed from the matroid of the original system. The complex-reflection/Shephard group examples recently constructed in \cite{L} would be a good place to start: these already have interesting symmetry groups automatically build into their construction.

The small orbit property also provides an explanation of the ad-hoc construction of elliptic $\bigvee$-systems \cite{St1} and elliptic solutions to the WDVV equations. These solutions have, as their leading term, a function that by itself is a solution to the WDVV equations of the form (\ref{rational}), but an irregular orbit had to be added, but which irregular orbit was unclear. It turns out that the irregular orbits are precisely small orbits. One observation coming from these results is that, for Coxeter groups, the existence of an irreducible quartic invariant polynomial is equivalent to the existence of a small orbit. Whether this is significant or just accidental is unclear.

The WDVV-equations and the rational solutions come from the commutativity, or zero-curvature relations, for the deformed connection
\[
\nabla_a = \partial_a + \kappa \sum_{\alpha\in\mathcal{U}} \frac{ (\alpha,a) }{(\alpha,z)} \alpha^\vee \otimes \alpha\,.
\]
The construction in this paper can also be thought of in terms of extending such a connection into a dimension higher. The geometry of such a construction also deserves to be studied. Such questions also appear in the Hurwitz space description. For example, consider the $A_n$-example and its extension. This construction corresponds to adding an extra term to the superpotential:
\[
\left.\prod_{i=0}^n (z-z^i)\right|_{\sum_{i=1}^n z^i=0}  \mapsto \left.\frac{ \prod_{i=0}^{n} (z-z_i) }{(z-z^\circ)^k}\right|_{\sum_{i=0}^{n} z_i - n z^\circ =0}
\]
and the geometry of Hurwitz theory requires that $k\in\pm\mathbb{N}\,.$ Algebraically this restriction is not required in the $\bigvee$-system.
It is here that the sign of $k$ effects the geometry (but not the algebra). If $k$ is positive this generates the extended affine Weyl orbit space and a solution that is almost dual to the corresponding Frobenius manifold. If $k$ is negative the superpotential no longer has a pole, but a multiple root. This corresponds to the induced Frobenius structures on discriminant surfaces with a larger manifold \cite{St2}. That such induced structures on discriminant generate solutions to the WDVV equations of the form (\ref{rational}) was proved in \cite{FV2}.

Finally, root systems and $\bigvee$-systems appear in many other places in mathematical physics: in the theory of Calogero-Moser and Schr\"odinger operators, for example \cite{SV}, and there are rational and trigonometric versions of both of these. Whether these are connected by a suitable Legendre transformation is unknown.

All these questions require further work.

\section*{Acknowledgements}

\noindent Richard Stedman was funded by the ESPRC Doctoral Training Grant EP/K503058/1.

\end{document}